\documentclass[journal,twoside,web]{ieeecolor}
\usepackage{generic}
\usepackage{cite}
\usepackage{amsmath,amssymb,amsfonts}
\usepackage{algorithmic}
\usepackage{graphicx}
\usepackage{textcomp}
\usepackage[scr=rsfso,cal=zapfc,frak=euler,bb=ams]{mathalfa}
\usepackage{enumerate}
\let\labelindent\relax
\usepackage{enumitem}
\usepackage{booktabs}

\newtheorem{theorem}{Theorem}
\newtheorem{lemma}[theorem]{Lemma}
\newtheorem{remark}{Remark}

\newtheorem{definition}{Definition}

\newtheorem{proposition}[theorem]{Proposition}

\DeclareSymbolFont{bbold}{U}{bbold}{m}{n}
\DeclareSymbolFontAlphabet{\mathbbold}{bbold}
\newcommand{\vect}[1]{\mathbbold{#1}}

\newcommand{\qed}{\hfill $\blacksquare$}
\newcommand{\Med}{\textup{Med}}

\newcommand{\M}{\mathcal{M}}
\newcommand{\N}{\mathcal{N}}

\newcommand{\G}{\mathcal{G}}
\newcommand{\V}{\mathcal{V}}
\newcommand{\U}{\mathcal{U}}
\newcommand{\E}{\mathcal{E}}

\newcommand{\Gdecisive}{\G_{\textup{decisive}}}

\graphicspath{{./fig/}}

\def\BibTeX{{\rm B\kern-.05em{\sc i\kern-.025em b}\kern-.08em
    T\kern-.1667em\lower.7ex\hbox{E}\kern-.125emX}}
\begin{document}
\title{Convergence, Consensus and Dissensus in the Weighted-Median Opinion Dynamics}
\author{Wenjun Mei, \IEEEmembership{Member, IEEE}, Julien M. Hendrickx, \IEEEmembership{Member, IEEE}, Ge Chen, \IEEEmembership{Member, IEEE}, Francesco Bullo, \IEEEmembership{Fellow, IEEE} and Florian D\"orfler, \IEEEmembership{Member, IEEE}
\thanks{This work was supported in part by the National Natural Science Foundation of China under Grants No. 72201008, No. 72131001, No. 72192804, and No. 12071465, the U.S. Army Research Office under Grant No. W911NF-22-1-0233, the RevealFlight Concerted Research Action (ARC) of the Federation Wallonie-Bruxelles, the Incentive Grant for Scientific Research (MIS) “Learning from Pairwise Data” of the F.R.S.-FNRS, as well as ETH Zurich funds.}
\thanks{W. Mei is with the Department of Mechanics and Engineering Science, Peking University (e-mail: mei@pku.edu.cn). J. M. Hendrickx is with the Institute of Information and Communication Technologies, Electronics, and Applied Mathematics, UCLouvain (e-mail: julien.hendrickx@uclouvain.be).G. Chen is with the Academy of Mathematics and Systems Science, Chinese Academy of Sciences (e-mail: chenge@amss.ac.cn). F. Bullo is with the Center of Control, Dynamical-Systems and Computation, University of California at Santa Barbara (e-mail: bullo@ucsb.edu). F. D\"orfler is with Automatic Control Laboratory, ETH Zurich (e-mail: dorfler@ethz.ch).}
\thanks{The corresponding author is Wenjun Mei (e-mail: mei@pku.edu.cn).}}

\maketitle

\begin{abstract}
Mechanistic and tractable mathematical models play a key role in understanding how social influence shapes public opinions. Recently, a weighted-median mechanism has been proposed as a new micro-foundation of opinion dynamics and validated via experimental data. Numerical studies also indicate that this new mechanism recreates some non-trivial real-world features of opinion evolution. In this paper, we conduct a thorough theoretical analysis of the weighted-median opinion dynamics. We fully characterize the set of all equilibria, and we establish the almost-sure finite-time convergence for any initial condition. Moreover, we prove a necessary and sufficient graph-theoretic condition for the almost-sure convergence to consensus, as well as a sufficient graph-theoretic condition for almost-sure persistent dissensus. It turns out that the weighted-median opinion dynamics, despite its simplicity in form, exhibit rich dynamical behavior that depends on some delicate network structures. To complement our sufficient conditions for almost-sure dissensus, we further prove that, given the influence network, determining whether the system almost surely achieves persistent dissensus is NP-hard, which reflects the complexity the network topology contributes to opinion evolution. 
\end{abstract}


\section{Introduction}\label{sec:introduction}
\subsection{Background and motivation}
Opinion dynamics study how individuals in groups change their opinions via interpersonal influence as well as the role of social network structures in shaping public opinions. Due to the complexity of social influence, a key to understanding opinion evolution in social groups relies on the construction of mechanistic and tractable mathematical models. A large class of opinion dynamics models have been proposed based on the classic French-DeGroot model~\cite{JRPF:56,MHDG:74}, which assumes that individuals update their opinions by taking some weighted averages of their social neighbors. 
However, as pointed out in a recent paper~\cite{WM-FB-GC-JH-FD:22}, the widely-adopted weighted-averaging mechanism features a non-negligibly unrealistic implication. That is, the attractiveness between any two individuals' opinions is linearly proportional to their opinion distance, which leads to consensus under mild network connectivity assumptions. 

In~\cite{WM-FB-GC-JH-FD:22}, the authors derive a new micro-foundation of opinion dynamics: the weighted-median mechanism. This new mechanism exhibits various desirable features: Empirical validation via a set of online experiment data indicates that, compared with the weighted-averaging mechanism, the new mechanism enjoys significantly lower errors in predicting individuals' opinion shifts driven by social influence. Numerical comparisons show that the weighted-median mechanism recreates some non-trivial features of real-world opinion evolution, which some widely-studied extensions of the French-DeGroot model do not fully capture. Those features include the decaying likelihood of reaching consensus with increasing group size or clustering coefficient; the emergence of various empirically observed public opinion distributions; and the pattern of how extreme opinions are located in social networks. Moreover, the weighted-median mechanism is independent of numerical representation of opinions but only requires them to be ordered. Therefore, it extends the applicability of opinion dynamics to multiple-choice issues with discrete and ordered options, e.g., political elections. 

Simulation studies indicate that opinion evolution via the weighted-median mechanism exhibits rich patterns dependent on some delicate network structures~\cite{WM-FB-GC-JH-FD:22}. However, its dynamical behavior remains to be rigorously analyzed. In terms of dynamical behavior, researchers on opinion dynamics primarily focus on the equilibrium set, the convergence of the systems, and graph-theoretic conditions for reaching consensus and persistent dissensus respectively. Typically, these properties are studied in the framework of consensus algorithms by leveraging certain network connectivity conditions.
Nevertheless, this framework does not apply to the analysis of the weighted-median opinion dynamics since more sophisticated network-structure properties rather than connectivity are involved. 

In this paper, we identify two important network structures that shape the asymptotic behavior of the weighted-median opinion dynamics: \emph{cohesive sets} and \emph{decisive links}. By establishing some important properties of these two structures, we conduct a thorough theoretical analysis of the weighted-median opinion dynamics, including fully characterizing the equilibrium set and establishing its almost-sure convergence in finite time. A necessary and sufficient graph-theoretic condition for asymptotic consensus is provided. Last but not least, by building connections between the weighted-median model and the monotone non-all-equal 3 satisfiability (NAE3SAT) problem, we prove that determining whether weighted-median opinion dynamics almost surely reach persistent dissensus is an NP-hard problem.


\subsection{Brief Review of Previous Averaging-Based Models}
The French-DeGroot model~\cite{JRPF:56,MHDG:74} is one of the earliest models of opinion dynamics. It assumes that individuals' opinions are denoted by real numbers and are updated by taking some weighted averages of their social neighbors' opinions. The interpersonal weights constitute an influence matrix, which in turn induces a directed and weighted graph called the influence network. The weighted-averaging mechanism implies overly large attractions between distant opinions, which drive the system to consensus under mild conditions. Essentially, the French-DeGroot model is a linear consensus algorithm. Spectral analysis of the interpersonal influence matrix indicates that The French-DeGroot model achieves asymptotic opinion consensus as long as the associated influence network has a globally reachable and aperiodic strongly connected component~\cite{FB:22}. This is an overly-simplified prediction since many real-world social systems, with their influence networks being connected, do not always achieve consensus.

To explain the prevalence of persistent dissensus, various important extensions have been proposed by introducing additional mechanisms. To name a few widely-studied representatives, the Friedkin-Johnsen (F-J) model~\cite{NEF-ECJ:90} assumes that individuals have persistent attachments to their initial opinions. The F-J model is a discrete-time linear system. Its convergence is established via matrix spectral analysis. In this model, a group reaches persistent dissensus if and only if at least two individuals start with distinct opinions. 

The Altafini model~\cite{CA:13} assumes the presence of negative weights in the influence network. This model is still a discrete-time linear system but the presence of negative weights adds some difficulty to spectral analysis. One way of analyzing such a system is to introduce an auxiliary ``lifted graph''~\cite{JMH:14}, which transforms the original Altafini model into a higher-dimensional linear system without negative coefficients. Spectral analysis of this augmented system indicates that the Altafini model achieves bipartite consensus if and only if the influence network is strongly connected and structurally balanced. For any strongly connected but structurally unbalanced network, all the individuals' opinions converge to the trivial state 0. 

The bounded-confidence models, including the synchronous Hegselmann-Krause (H-K) model~\cite{RH-UK:02} and the gossip-like Deffuant-Weisbuch (D-W) model~\cite{GD-DN-FA-GW:00}, assume that individuals only assign weights to opinions within certain distances from their own opinions. Such discontinuous truncation of social influence makes bounded-confidence models very challenging to analyze. In the case of homogeneous individuals, convergence and convergence rate are thoroughly studied in the framework of linear consensus algorithms with time-varying topologies~\cite{JD:01,AB-MB-BC-HLN:13,JZ-YH:13}. The behavior of the H-K model with heterogeneous individuals in general ``remain a mystery''~\cite{BC-CW:17}. Only the convergence under some specific conditions or with noises has been established~\cite{SRE-TB:15,AM-FB:11f,GC-WS-SD-YH:19}. For the heterogeneous D-W model, Chen et al.~\cite{GC-WS-WM-FB:18n} establish its almost-sure exponential convergence for certain range of model parameters by leveraging the technique of "transforming randomness into control inputs"~\cite{GC:17b}, which is also used in the analysis of the weighted-median opinion dynamics in this paper. 

The biased-assimilation model~\cite{PD-AG-DTL:13} introduces a highly non-linear modification of the French-DeGroot model to characterize the effect that individuals weigh confirming information more than dis-confirming information. Conditions for the convergence to polarization, persistent dissensus, or consensus are analyzed in~\cite{PD-AG-DTL:13} for specific types of networks. Some local stability and attractivity properties are established in~\cite{WX-MY-JL-MC-XS:20}.

To sum up, the models reviewed above exhibit one of two possible behaviors: almost-sure consensus or almost-sure dissensus, dependent on one specific condition, or they introduce highly nonlinear assumptions with the side effect that their convergence and consensus conditions become mathematically intractable. We refer to~\cite{AVP-RT:17,AVP-RT:18} for an insightful survey of recent progress in the modeling and analysis of opinion dynamics.

\subsection{Contribution}
As indicated by simulations in~\cite{WM-FB-GC-JH-FD:22}, the weighted-median opinion dynamics often lead to opinion clustering and exhibit rich dynamical behavior regarding the conditions for reaching consensus or dissensus. In this paper, we conduct a thorough analysis of the weighted-median opinion dynamics~\cite{WM-FB-GC-JH-FD:22} and identify two important network structures that shape the opinion evolution. The contributions of this paper include the following aspects.

Firstly, we fully characterize the set of all the equilibria. It turns out that the equilibria of the weighted-median opinion dynamics exhibit a clear pattern that depends on an important delicate structure of the influence network: the cohesive sets. The notion of cohesive set is first proposed by Morris~\cite{SM:00} and is widely adopted in the analysis of linear threshold models of network diffusion, e.g., see~\cite{EY-DA-AO:11}. In this paper, we adopt a special case in~\cite{SM:00} as the definition of cohesive sets, and provide some useful properties of it.

Secondly, we establish the almost-sure finite-time convergence of the weighted-median opinion dynamics, with respect to initial conditions and individual update sequences. Then we provide a necessary and sufficient graph-theoretic condition for the almost-sure convergence to consensus, and give a sufficient graph-theoretic condition for the convergence to almost-sure persistent dissensus. Note that almost-sure persistent dissensus is not automatic in the absence of almost-sure consensus, as there is a middle ground where different outcomes have positive probabilities. Theoretical analysis results indicate that, under various network structure conditions, the weighted-median model either almost surely converges to consensus, or almost surely converges to persistent dissensus, or has a non-zero probability of reaching dissensus, depending on the initial conditions. In addition, network connectivity conditions are not sufficient to guarantee consensus. Therefore, the weighted-median opinion dynamics are less likely to reach consensus and exhibit richer dynamical behavior than the French-DeGroot model.

Thirdly, we show that, given an influence network, determining whether the weighted-median opinion dynamics almost surely achieve persistent dissensus is equivalent to determining that for a finite set of initial conditions. We then further prove that the latter problem is NP-hard by relating it to the monotone non-all-equal 3 satisfiability (NAE3SAT) problem, which is known to be NP-hard. This result reflects the complexity that network topology contributes to the dynamical behavior of the weighted-median model, which highlights another important difference with averaging-based consensus problems and shows an arguably more benign behavior.

\subsection{Organization}
The rest of this paper is organized as follows. Section II introduces some basic definitions and notions, as well as the model setup of the weighted-median opinion dynamics. Section III presents all the theoretical analysis and proofs of the main results. Section IV is the conclusion. Proofs of lemmas are provided in the appendices.

\section{Basic Definitions and Model Setup}
Let $\subseteq$ and $\subset$ be the symbols for subset and proper subset respectively. Denote by $\mathbb{N}$ the set of natural numbers, i.e., $\mathbb{N}=\{0,1,2,\dots\}$. Let $\mathbb{Z}$ (resp. $\mathbb{Z}_+$) be the set of (resp. positive) integers. Let $\vect{1}_n$ and $\vect{0}_n$ be the $n$-dimension vector whose entries are all ones and all zeros respectively.

Denote by $\G(W)$ the directed and weighted graph associated with the adjacency matrix $W$. In this paper we use the terms ``graph'' and ``network'' interchangeably. Suppose there are $n$ nodes on the graph, i.e., $W=(w_{ij})_{n\times n}$. Let $\V = \{1,\dots,n\}$ be the index set of the nodes. Denote by $\N_i$ the set of node $i$'s out-neighbors, i.e., $\N_i = \{j\in \mathcal{V}|w_{ij}\neq 0\}$, which includes node $i$ itself if $w_{ii}\neq 0$. A network $\G(W)$ is referred to as an \emph{influence network} if the nodes on $\G(W)$ represent individuals and any $(i,j)$-entry of the \emph{influence matrix} $W$ represents how much individual $i$ is influenced by $j$. Conventionally, an influence matrix $W$ is assumed to be row-stochastic.

The formal definition of weighted median is given below.
\begin{definition}[Weighted median~\cite{WM-FB-GC-JH-FD:22}]\label{def:weighted-median-in-general}
Given any $n$-tuple of real values $x=(x_1,\dots, x_n)$ and any $n$-tuple of non-negative weights $w=(w_1,\dots, w_n)$ with $\sum_{i=1}^n w_i=1$, $x^*\in \{x_1,\dots, x_n\}$ is a weighted median of $x$ associated with the weights $w$ if $x^*$ satisfies 
   \begin{align*}
   \sum_{i:\, x_i<x^*} w_i \le 1/2,\quad \text{and} \quad \sum_{i:\, x_i>x^*}w_i \le 1/2.
   \end{align*}
\end{definition} 
\smallskip
For simplicity, we also say that $x^*$ is a weighted median of $x$ associated with $w$. 

The following lemma rephrases Appendix~A in~\cite{WM-FB-GC-JH-FD:22}. It presents some immediate results regarding the uniqueness of weighted median.
\begin{lemma}[Properties of weighted median]\label{lem:properties-weighted-median}
Given any $n$-tuple of real values $x=(x_1,\dots,x_n)$ and any $n$-tuple of the associated non-negative weights $w=(w_1,\dots, w_n)$ with $\sum_{i=1}^n w_i=1$, let $x_{(1)},x_{(2)},\dots,x_{(n)}$ be a re-ordering of $x_1,\dots, x_n$ such that $x_{(1)}\le x_{(2)}\le \dots \le x_{(n)}$. The weighted median of $x$ associated with $w$ is unique if and only if there exists $1<i^*<n$ such that 
\begin{align*}
\sum_{i=1}^{i^*-1} w_{(i)} < \frac{1}{2},\quad w_{(i^*)}>0,\quad \sum_{i=i^*+1}^n w_{(i)}<\frac{1}{2}.
\end{align*}
In this case, $x_{(i^*)}$ is the unique weighted median of $x$ associated with $w$. When such $i^*$ does not exists, there exist $1<\underline{i}<\overline{i}<n$ such that
\begin{align*}
\sum_{i=1}^{\underline{i}-1} w_{(i)}<\frac{1}{2},\,\,\, \sum_{i=1}^{\underline{i}}w_{(i)}=\sum_{i=1}^{\overline{i}} w_{(i)}=\frac{1}{2},\,\,\, \sum_{i=\overline{i}+1}^n w_{(i)} < \frac{1}{2},
\end{align*}
which also implies that $w_{(\underline{i}+1)}=\dots=w_{(\overline{i}-1)}=0$. In this case, $x_{\underline{i}},\, x_{\underline{i}+1},\, \dots, x_{\overline{i}}$ are all weighted medians of $x$ associated with $w$.
\end{lemma}


The weighted-median opinion dynamics have been proposed in~\cite{WM-FB-GC-JH-FD:22} as a discrete-time stochastic process, in which one individual updates their opinion at each time step.
\begin{definition}[Weighted-median opinion dynamics~\cite{WM-FB-GC-JH-FD:22}]\label{def:WM-op-dyn}
Consider a group of $n$ individuals in an influence network associated with a row-stochastic influence matrix $W$. For any $i\in V$ and any $t\in \mathbb{N}$, denote by $x_i(t)$ individual $i$'s opinion at time $t$. The weighted-median opinion dynamics is defined as the following stochastic process: At each time $t+1$, one individual $i$ is uniformly randomly picked and updates their opinion according to the following equation:
\begin{equation*}
x_i(t+1) = \Med_i\big(x(t);W\big),
\end{equation*}
where $\Med_i(x(t);W)$ is the weighted median of $x(t)$ associated with the weights given by the $i$-th row of $W$, i.e., $(w_{i1},w_{i2},\dots, w_{in})$. If the weighted-median is not unique, then let $\Med_i\big(x(t);W\big)$ be the weighted median that is the closest to $x_i(t)$.
\end{definition}

\begin{remark}
The model setup in Definition~\ref{def:WM-op-dyn} guarantees that $\Med_i\big(x(t);W\big)$ is always unique for any $i$ and $t$, see~\cite{WM-FB-GC-JH-FD:22} for a detailed discussion. As shown in~\cite{WM-FB-GC-JH-FD:22}, the weighted-median opinion dynamics can be interpreted as a best-response dynamics, in which individuals update their opinions by taking the optimal solutions that myopically minimize their social pressure caused by disagreeing with others, i.e.,
\begin{align*}
x_i(t+1) \in \text{argmin}_{z\in \mathbb{R}} \sum_{j=1}^n w_{ij}|z-x_j(t)|.
\end{align*}
\end{remark}

\begin{remark}[Connections with other models]
According to Lemma~\ref{lem:properties-weighted-median}, the weighted median is a self-map. Hence, any individual $i$ in the weighted-median opinion dynamics can only adopts opinions $x_i(t)$ that initially exist in the system, i.e., $x(t)\in \{x_1(0),\dots, x_n(0)\}^n$ for any $t\in \mathbb{N}$. Therefore, for any fixed initial condition, the weighted-median model given by Definition~\ref{def:WM-op-dyn} is essentially a discrete-time Markov chain. In the case when there are only two distinct initial opinions, the weighted-median model coincides with several well-studied network dynamics, e.g., the deterministic voter model based on the local majority rule~\cite{RAH-TML:75,BCB:97} and the linear threshold model with homogeneous threshold 1/2~\cite{EY-DA-AO:11,AG-WL-LVSL:11}. If the underlying influence network is a lattice, then it becomes the Glauber model of social segregation~\cite{RM:11}. However, the two-opinion scenario does not reflect a key feature of the weighted-median opinion dynamics: the discrete but ordered opinions set.
\end{remark}

\section{Dynamical behavior of the weighted-median opinion dynamics}

It is widely believed that influence network structure plays an important role in shaping opinion evolution in social groups. However, in some widely-studied averaging-based models, e.g., the French-DeGroot model and the bounded-confidence models, whether a system reaches consensus is only determined by the connectivity of the (potentially time-varying) influence network. The effects of other finer network structures are not fully unveiled. The intuition behind this problem is that the weighted-averaging mechanism implies too large ``attractive forces'' between distant opinions. As a result, no structural property other than the lack of connectivity is capable of harnessing the system from converging to consensus.

On the other hand, the weighted-median mechanism resolves the above issue by assuming independence between opinion attractions and opinion distances. With this fundamental change, the weighted-median opinion dynamics exhibit some new features in terms of how influence network structure shapes opinion evolutions. Fig.~\ref{fig:eg-lattice}, taken from~\cite{WM-FB-GC-JH-FD:22}, shows a typical example of a single simulation of the weighted-median model on a lattice graph. One could observe that the system does not reach a consensus at the equilibrium. Moreover, opinions might form local clusters, which indicates that some local structures of the influence network might play roles in shaping the system's asymptotic behavior.

In this section, we identify two such structures and analyze how they affect the weighted-median model's dynamical behavior. These two structures are cohesive sets and decisive links. We provide some important properties of them, based on which we characterize the set of equilibria, establish its convergence, and propose conditions for reaching asymptotic consensus and dissensus respectively.

\begin{figure}
\centering
\includegraphics[width=1\linewidth]{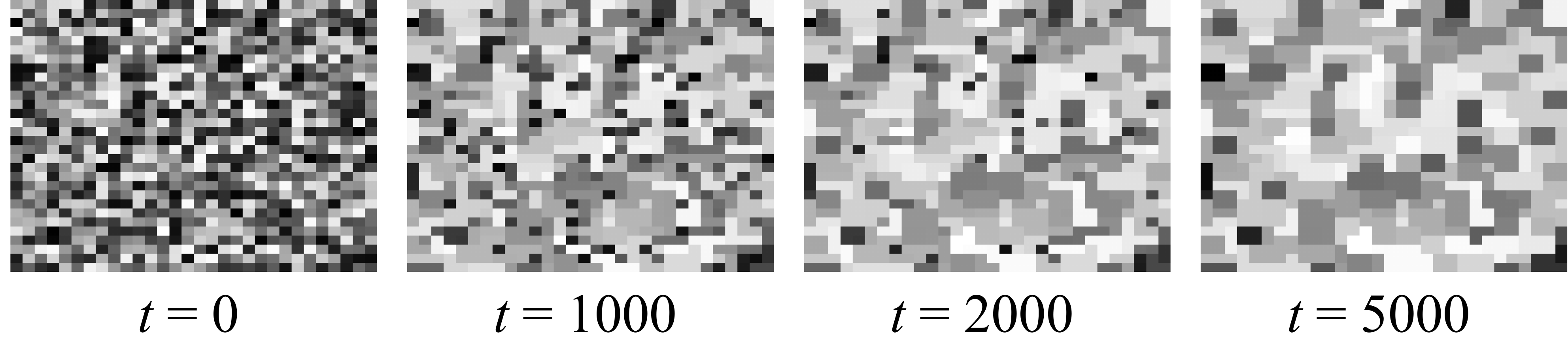}
\caption{One simulation of the weighted-median model on a 30$\times$30 lattice graph. This figure is taken from~\cite{WM-FB-GC-JH-FD:22}, where it is Fig.~3. In this figure, each block is an individual bilaterally connected with all their adjacent blocks (not including the diagonally adjacent blocks). Each individual has a self-loop and uniformly assigns weights to all their neighbors including themselves. Initial opinions are independently randomly generated according to the uniform distribution on $[-1,1]$. The gray scale of each block is proportional to the absolute value of the individual's final opinion. After 5000 time steps, the evolution reaches an equilibrium.}\label{fig:eg-lattice}
\end{figure}

\subsection{Important concepts and lemmas}

The concept of cohesive sets was first proposed in~\cite{SM:00} and used in the study of linear-threshold models of network diffusion~\cite{DA-AO-EY:11}, with a generalized form. In this paper, we adopt a specific version of it. See Fig.~\ref{fig:concepts-CS-DL}(a) for a visualized example.
\begin{definition}[Cohesive set and maximal cohesive set]\label{def:CS-MCS}
Given an influence network $\mathcal{G}(W)$ with node set $V$, a cohesive set $\mathcal{M}\subset \mathcal{V}$ is a subset of nodes that satisfies $\sum_{j\in \mathcal{M}} w_{ij}\ge 1/2$ for any $i\in \mathcal{M}$. A cohesive set $\M$ is a maximal cohesive set if there does not exists $i\in \V\setminus \M$ such that $\sum_{j\in \M} w_{ij}>1/2$.
\end{definition}

Note that, in the weighted-median opinion dynamics, if all the nodes in a cohesive set adopt the same opinion, then that opinion will also be the weighted-median opinion for each individual in the cohesive set. Namely, these individuals' opinions will not be updated via the weighted-median mechanism. In this sense, the concept of cohesive set can be interpreted as a characterization of ``echo chamber'' (a metaphorical description of a situation in which beliefs are amplified by communication and repetition inside a closed system).

Given the concept of cohesive set, we further define a process called \emph{cohesive expansion}.
\begin{definition}[Cohesive expansion]\label{def:cohesive-expansion}
Given an influence network $\G(W)$ with node set $\V$ and a subset of nodes $\M\subset \V$, a subset of $\V$ is a cohesive expansion of $\M$ if it can be constructed via the following iteration algorithm:
\begin{enumerate}
\item Let $\M_0=\M$;
\item For $k=0,1,2,\dots$, if there exists $i\in \V\setminus \M_k$ such that $\sum_{j\in \M_k} w_{ij}>1/2$, then let $\M_{k+1}=\M_k \cup \{i\}$;
\item Terminate the iteration at step $k$ as long as there does not exists any $i\in \V\setminus \M_k$ such that $\sum_{j\in \M_k} w_{ij}>1/2$, and let $\textup{Expansion}(\M)=\M_k$.
\end{enumerate}
\end{definition}

In what follows, we present and prove some novel results on the fundamental properties of cohesive sets, cohesive expansion, and maximal cohesive sets. 

The following lemma states that cohesive expansion is unique. As a result, we can denote by $\textup{Expansion}(\M)$ the cohesive expansion of a node subset $M$ without ambiguity. The proof of Lemma~\ref{lem:uniqueness-cohesive-expansion} is given in Appendix~\ref{append:proof-uniqueness-cohesive-expansion}.
\begin{lemma}[Uniqueness of cohesive expansion]\label{lem:uniqueness-cohesive-expansion}
Given an influence network $\mathcal{G}(W)$ with the node set $\V$, for any $\M\subset \V$, the cohesive expansion of $\M$ is unique, i.e., independent of the order of node additions. 
\end{lemma}

The following lemma presents some fundamental properties of cohesive sets and cohesive expansions. The proof is given in Appendix~\ref{append:proof-lem:properties-cohesive-expansion}.
\begin{lemma}[Properties of cohesive sets/expansions]\label{lem:properties-cohesive-expansion}
Given an influence network $\mathcal{G}(W)$ with the node set $\V$, the following statements hold:
\begin{enumerate}
\item If $\M_1,\,\M_2\subseteq \V$ are both cohesive sets, then $\M_1\cup \M_2$ is also a cohesive set;
\item For any $\M,\, \tilde{\M}\subset \V$, if $\M\subseteq \tilde{\M}$, then $\textup{Expansion}(\M)\subseteq \textup{Expansion}(\tilde{\M})$;
\item For any $\M,\, \tilde{\M}\subset \V$, $\textup{Expansion}(\M)\cup \textup{Expansion}(\tilde{\M})$ $\subseteq$ $\textup{Expansion}(\M\cup \tilde{\M})$; and
\item If $\M$ is a cohesive set, then $\textup{Expansion}(\M)$ is also cohesive and is the smallest maximal cohesive set that includes $\M$, that is, for any maximal cohesive set $\hat{\M}$ such that $\M\subset \hat{\M}$, we have $\textup{Expansion}(\M)\subset \hat{\M}$. 
\end{enumerate}
\end{lemma}

Below we present another useful lemma on cohesive sets, which is a straightforward consequence of Definition~\ref{def:CS-MCS} and Lemma~\ref{lem:properties-cohesive-expansion}~4).
\begin{lemma}[Cohesive partition]\label{lem:cohesive-partition}
Given an influence network $\mathcal{G}(W)$ with node set $\V$ and a cohesive set $\M\subset \V$, if $\M$ is maximally cohesive, then $\V\setminus \M$ is also maximally cohesive; If $\M$ is not maximally cohesive, then either of the following two statements holds:
\begin{enumerate}
\itemsep0em 
\item $\textup{Expansion}(\M)=\V$;
\item $\textup{Expansion}(\M)$ and $\V\setminus \textup{Expansion}(\M)$ are both non-empty and maximally cohesive. 
\end{enumerate}
\end{lemma}

Another important concept involved in called \emph{decisive link}. See Fig.~\ref{fig:concepts-CS-DL}(b)(c) for visualized illustrations. Loosely speaking, a link $(i,j)$ is decisive if node $j$ could potentially be a ``tie-breaker'' when $i$ computes their weighted-median opinion.
\begin{definition}[Decisive and indecisive out-links]\label{def:decisive-indecisive-links}
Given an influence network $\mathcal{G}(W)$ with the node set $\V$, define the out-neighbor set of each node $i$ as $\N_i = \{j\in \V\,|\, w_{ij}\neq 0\}$. A link $(i,j)$ is a decisive out-link of node $i$, if there exists a subset $\theta\subset \N_i$ such that the following three conditions hold: (1) $j\in \theta$; (2) $\sum_{k\in \theta} w_{ik} > 1/2$; (3) $\sum_{k\in \theta\setminus \{j\}} w_{ik}<1/2$. Otherwise, the link $(i,j)$ is an indecisive out-link of node $i$.
\end{definition}

\begin{remark}\label{remark:implication-decisive/indecisive}
Definition~\ref{def:decisive-indecisive-links} together with Definition~\ref{def:weighted-median-in-general} indicates that individuals in the weighted-median opinion dynamics adopt only opinions of their out-neighbors linked to via decisive links. In addition, if $(i,j)$ is an indecisive link, then individual $i$'s weighted-median opinion $\Med_i(x;W)$ never depends on individual $j$'s opinion $x_j$.
\end{remark}

\begin{figure}
\centering
\includegraphics[width=0.95\linewidth]{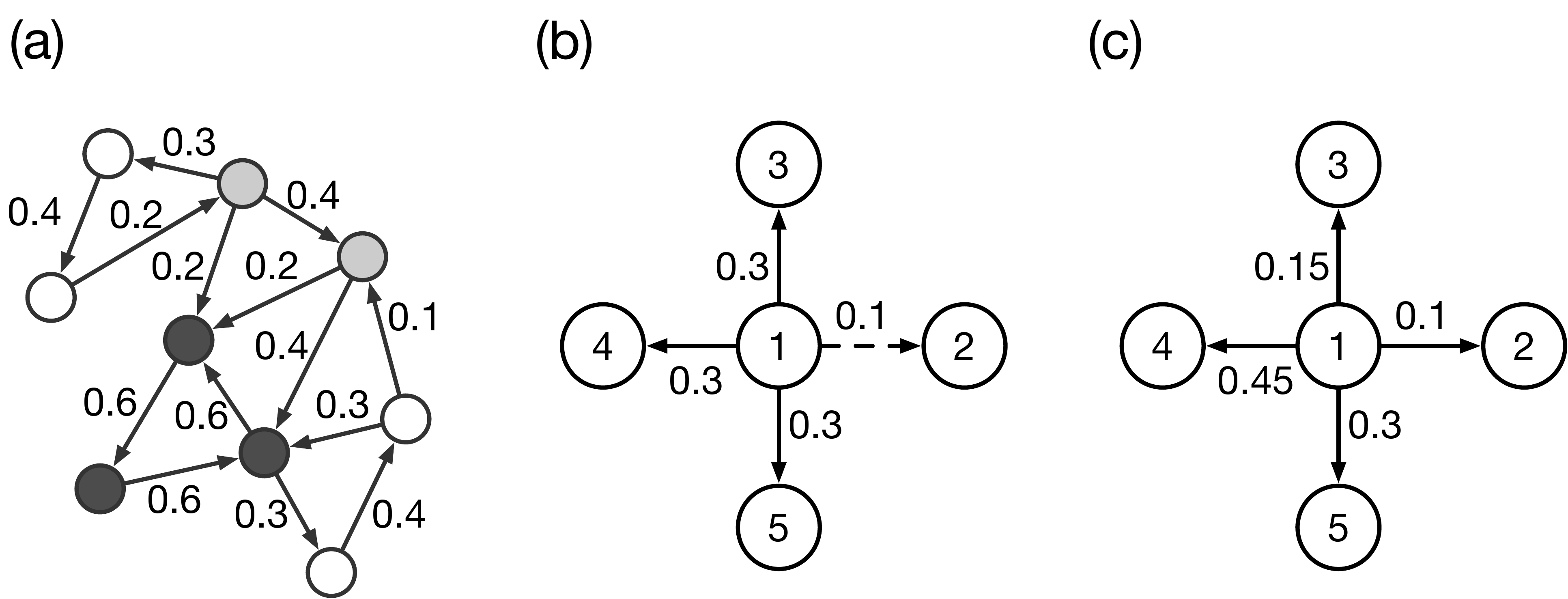}
\caption{Visualized examples illustrating the concepts of cohesive sets, maximal cohesive sets, decisive links and indecisive links. In Panel~(a), the dark grey nodes constitute a cohesive set. The dark grey and the light grey nodes together form a maximal cohesive set. In Panel~(b), the link $(1,2)$ is indecisive, while the other links are all decisive. In Panel~(c), all the links are decisive.}\label{fig:concepts-CS-DL}
\end{figure}

\subsection{Equilibria, convergence, consensus and persistent dissensus}
In this subsection, we present and prove our main results on the dynamical behavior of the weighted-median opinion dynamics. First of all, the theorem below characterizes the set of all the equilibria.
\begin{theorem}[Set of equilibria]\label{thm:equilibrium-set}
Given an influence network $\G(W)$ with $n$ individuals, $x^*\in \mathbb{R}^n$ is an equilibrium of the weighted-median opinion dynamics defined by Definition~\ref{def:WM-op-dyn} if and only if $x^*$ is either a consensus vector, i.e., $x_i^*=x_j^*$ for any $i$, $j\in \V$, or satisfies the following condition: for any $y$ satisfying $\min_i x_i^* < y \le \max_i x_i^*$, both $\{i\in \V\,|\, x_i^*<y\}$ and $\{i\in \V\,|\, x_i^*\ge y\}$ are maximal cohesive sets on $\G(W)$.
\end{theorem}

\begin{proof}
We first prove the ``if'' part. If $x^*$ is a consensus vector, then we have $\Med_i(x^*;W)=x_i^*$ for any $i$, since the weighted-median of only one opinion is always the opinion itself. Now suppose that $x^*$ is NOT a consensus vector and satisfies that, for any $y\in (\min_k x_k^*, \max_k x_k^*)$, $\{j\in \V|x_j^* <y\}$ and $\{j\in \V|x_j^*\ge y\}$ are both maximal cohesive sets. For any given $i$, let $y=x_i^*$. Since $i\in \{j:\, x_j^*\ge x_i^*\}$, and $\{j:\, x_j^*\ge x_i^*\}$ is a maximal cohesive set, we have
\begin{equation*}
\sum_{j:\, x_j^*\ge x_i^*} w_{ij}\ge \frac{1}{2} \quad \Rightarrow \quad \Med_i(x^*;W)\ge x_i^*.
\end{equation*}
Let $\tilde{y}=\min_k \{x_k^*\,|\, x_k^*>x_i^*\}$. Since $i\in \{j:\, x_i^*< \tilde{y}\}$, and $\{j:\, x_i^*< \tilde{y}\}$ is a maximal cohesive set, we have
\begin{equation*}
\sum_{j:\, x_j^*< \tilde{y}} w_{ij}\ge \frac{1}{2} \quad \Rightarrow \quad \Med_i(x^*;W)<\tilde{y}.
\end{equation*}
The inequality $x_i^*\le \Med_i(x^*;W) <\tilde{y}$ together with $\Med_i(x^*;W)\ge x_i^*$ leads to $x_i^*= \Med_i(x^*;W)$. This concludes the proof for the ``if'' part.

Now we proceed to prove the ``only if'' part. Suppose $x^*$ is not a consensus vector and there exists $y\in (\min_k x_k^*, \max_k x_k^*)$ such that $\{j:\,x_j^* <y\}$ and $\{j:\,x_j^*\ge y\}$ are not both maximal cohesive sets. Since these two sets form a disjoint partition of the node set $\{1,\dots, n\}$, one of them must not be cohesive. Otherwise, by definition, we must have
\begin{align*}
\sum_{j:\, x_j^*<y} w_{pj} & =1-\!\sum_{j:\, x_j^*\ge y} w_{pj}\le \frac{1}{2},\text{ for any }p\in \{j:\, x_j^*\ge y\},\\
\sum_{j:\, x_j^*\ge y} w_{pj} & = 1-\! \sum_{j:\, x_j^*<y} w_{pj}\le \frac{1}{2},\text{ for any }p\in \{j:\, x_j^*<y\},
\end{align*}
which contradicts the pre-assumption that $\{j:\, x_j^*<y\}$ and $\{j:\, x_j^*\ge y\}$ are not both maximal cohesive sets.

Suppose $\{j:\, x_j^*\ge y\}$ is not cohesive. As a direct consequence, there exists $i$ with $x_i^*\ge y$ but 
\begin{equation*}
\sum_{j:\, x_j^*<y} w_{ij} > \frac{1}{2},
\end{equation*}
which in turn implies that $\Med_i(x^*;W)<y\le x_i^*$. Therefore, such $x^*$ is not an equilibrium of the weighted-median opinion dynamics. Similarly, if $\{j:\, x_j^*<y\}$ is not cohesive, then there exists $i$ with $x_i^*<y$ such that $\Med_i(x^*;W)\ge y>x_i^*$. That is, $x^*$ is not an equilibrium. Therefore, in order for $x^*$ to be an equilibrium, it must be either a consensus vector or $\{j:\,x_j^* <y\}$ and $\{j:\,x_j^*\ge y\}$ are both maximal cohesive sets for any $y\in (\min_k x_k^*, \max_k x_k^*)$. This concludes the proof for the ``only if'' part.
\end{proof}

If an equilibrium $x^*$ is not a consensus vector, we refer to it as a \emph{dissensus equilibrium}.

In what follows, we prove that the weighted-median opinion dynamics, with any initial condition, almost surely converge to an equilibrium in finite time. (Here ``almost surely'' is in terms of the random initial condition and the random individual update sequence given by Definition~\ref{def:WM-op-dyn}.) Moreover, we propose and prove necessary and sufficient graph-theoretic conditions for almost-sure convergence to consensus, as well as a sufficient graph-theoretic condition for almost-sure convergence to dissensus equilibria. 

The following lemma is a foundation of the proof for convergence and is akin to results about absorbing states in Markov chains. Its proof is provided in Appendix~\ref{append:proof-lem:monkey-typewriter}.
\begin{lemma}[Transforming randomness to sequence design]\label{lem:monkey-typewriter}
Consider the weighted-median opinion dynamics given by Definition~\ref{def:WM-op-dyn}. If, starting from any $x$, there exists an update sequence $i_1,\dots,i_{T_x}$, along which the opinion trajectory reaches an equilibrium at time step $T_x$, then the weighted-median opinion dynamics almost surely converges to an equilibrium in finite time, for any initial condition $x(0)$.
\end{lemma}

With Lemma~\ref{lem:monkey-typewriter}, we can establish the convergence of the weighted-median opinion dynamics by manually inputting to the system the ordering of individuals' opinion updates. That is, Lemma~\ref{lem:monkey-typewriter} provides a way of transforming randomness into control inputs. Similar approaches have been used in previous literature on synchronization and consensus analysis, e.g., see~\cite{AS-SET-VDB-RS:08,GC:17b,GC-WS-WM-FB:18n}.

Now we are ready to present the main theorem of this subsection. The core of its proof is to establish the convergence of the weighted-median model. We prove the convergence by first considering a two-opinion scenario and then extending the argument to general cases where there are more than two distinct opinions initially.
\begin{theorem}[Convergence, consensus, and dissensus]\label{thm:dynamics-WM}
Consider the weighted-median opinion dynamics given by Definition~\ref{def:WM-op-dyn}, on an influence network $\G(W)$ with node set $\V$. Denote by $\Gdecisive(W)$ the subgraph of $\G(W)$ with all the indecisive out-links removed. The following statements hold:\begin{enumerate}
\item For any initial condition $x(0)\in \mathbb{R}^n$, the solution $x(t)$ almost surely converges to an equilibrium $x^*$ in finite time;
\item If the only maximal cohesive set of $\G(W)$ is $\V$, then, for any initial condition $x(0)\in \mathbb{R}^n$,  the solution $x(t)$ almost surely converges to a consensus state;
\item If the graph $\G(W)$ has a maximal cohesive set $\M\neq \V$, then there exists a subset of initial conditions $X_0\subseteq \mathbb{R}^n$ with non-zero Lebesgue measure in $\mathbb{R}^n$ such that, for any $x(0)\in X_0$, there is no update sequence along which the solution converges to consensus; and
\item If $\Gdecisive(W)$ does not have a globally reachable node, then, for any initial condition $x(0)\in \mathbb{R}^n\setminus \hat{X}$, where $\hat{X}=\{x\in \mathbb{R}^n\,|\, \exists\, i\neq j \text{ s.t. }x_i=x_j\}$ has zero Lebesgue measure in $\mathbb{R}^n$, the solution $x(t)$ almost surely reaches a dissensus equilibrium in finite time.
\end{enumerate} 
\end{theorem}

\begin{proof}
Regarding statement~1), we first point out that it is equivalent to the following claim:
\begin{itemize}
\item Statement~1*): For any initial state $x(0)$, there exists an update sequence $\{i_1,\dots,i_T\}$ along which the solution $x(t)$ reaches an equilibrium at time step $T$.
\end{itemize}

``Statement~1) $\Rightarrow$ Statement~1*)'' is straightforward and ``1*) $\Rightarrow$ 1)'' is an immediate result of Lemma~\ref{lem:monkey-typewriter}. Therefore, in order to prove statement~1), we only need to prove~1*).

Now we prove that statement~1*) is true. We first consider the case in which there are only two different opinions initially in the network. Without loss of generality, let the two opinions be $y_1$ and $y_2$. Let
\begin{equation*}
\V_1(t) = \{i\in \V\,|\, x_i(t)=y_1\},\quad \V_2(t) = \{i\in \V\,|\, x_i(t)=y_2\}, 
\end{equation*}
for any $t\in \mathbb{N}$. Due to the weighted-median update rule given by Definition~\ref{def:WM-op-dyn}, for any initial state $x(0)\in \{y_1,y_2\}^n$, the solution $x(t)$ satisfies $x(t)\in \{y_1,y_2\}^n$ for any $t\ge 0$. Therefore, at any $t$, the sets $\V_1(t)$ and $\V_2(t)$ form a partition of the node set $\V$. We neglect the trivial cases when $\V_1(0)=\V$ or $\V_2(0)=\V$, otherwise the system would be already at an equilibrium. We construct an update sequence as follows:
\begin{enumerate}
\item For any time step $t+1$, $t=0,1,2,\dots$, if there exists some $i_{t+1} \in \V_1(t)$ such that $\sum_{j\in \V_2(t)}w_{i_{t+1}j}>1/2$, then update node $i_{t+1}$ at time step $t+1$ and thereby we get $\V_1(t+1)=\V_1(t)\setminus \{i_{t+1}\}$ and $\V_2(t+1)=\V_2(t)\cup \{i_{t+1}\}$;
\item The update stops at time step $T$ if there does not exists any $i\in \V_1(T)$ such that $\sum_{j\in \V_2(T)} w_{ij}>1/2$.
\end{enumerate}
By updating the system along the sequence $\{i_1,\dots,i_T\}$ we obtain a partition $\V_1(T)$ and $\V_2(T)$, and all the individuals in $\V_1(T)$ ($\V_2(T)$ resp.) hold the opinion $y_1$ ($y_2$ resp.). Note that $\V_2(T)$ is the cohesive expansion of $\V_2(0)$. However, since $\V_2(0)$ is not necessarily cohesive, $\V_2(T)$ is not necessarily cohesive either. 

If $\V_1(T)$ is empty, then the system is already at an equilibrium where all the nodes hold opinion $y_2$. If $\V_1(T)$ is not empty, then, for any $i\in \V_1(T)=\V\setminus \V_2(T)$, since $\V_2(T)$ is already the cohesive expansion of $\V_2(0)$, we have $\sum_{j\in \V_2(T)}w_{ij}\le 1/2$, which implies that
\begin{equation*}
\sum_{j\in \V_1(T)}\!w_{ij} = \sum_{j\in \V\setminus \V_2(T)}\!w_{ij} = 1 - \!\sum_{j\in \V_2(T)}\! w_{ij} \ge 1/2. 
\end{equation*}
Therefore, $\V_1(T)$ is cohesive. Denote by $\E_1 = \V_1(T) \cup (j_1,\dots,j_k)$ the cohesive expansion of $\V_1(T)$, and the nodes are added to $\V_1(T)$ along the sequence $j_1,\dots,j_k$. Now we construct the update sequence as $i_1,\dots,i_T,j_1,\dots,j_k$. If $\E_1 = \V$, then the system along this update sequence reaches the equilibrium where all the nodes adopt opinion $y_1$. If $\E_1\neq \V$, then the system along such an update sequence reaches the state in which all the nodes in $\E_1$ adopt opinion $y_1$ while all the nodes in $\V\setminus \E_1$ adopt opinion $y_2$. According to Lemma~\ref{lem:cohesive-partition}, $\E_1$ and $\V\setminus \E_1$ are both maximally cohesive sets. Therefore, the system reaches an equilibrium along the update sequence $i_1,\dots,i_T,j_1,\dots,j_k$. Till now we have proved statement~1*) in the case when there are only two distinct opinions initially.

Now we extend the above argument to the case of any arbitrary initial condition $x(0)\in \mathbb{R}^n$ by induction. Suppose statement~1*) holds whenever there are less than or equal to $R-1$ distinct opinions initially, for some $2<R\le n$. Suppose that there are $R$ distinct values among the initial opinions $\{x_1(0),\dots, x_n(0)\}$ and denote by these $R$ values $y_1,\,y_2\,\dots,\, y_R$, with $y_1<y_2<\dots<y_R$.

Let $A_1=\{y_1\}$ and $B_1=\{y_2,\dots, y_R\}$, and consider $A_1$, $B_1$ as two possible states of the nodes. Due to the weighted-median mechanism, whether a node converts its state from $A_1$ to $B_1$ depends only on which of its neighbors are in state $B_1$. Therefore, when discussing the nodes' state transitions from $A_1$ to $B_1$, we can assume that the nodes in state $B_1$ all have an identical opinion $\tilde{y}_1>y_1$. That is, the process of the state transitions from $A_1$ to $B_1$ is reduced to the two-opinion scenario, which we have previously discussed. Repeating the argument for the two-opinion case, we can construct an update sequence $i_{11},\dots,i_{1k_1}$ such that, at time $k_1$, either no individual hold the opinion $y_1$, or the node set $\V$ at time $k_1$ is divided into two sets $\E_1$ and $\V\setminus \E_1$, such that
\begin{enumerate}
\item all the nodes in $\E_1$ hold the opinion $y_1$;
\item $\E_1$ is a maximal cohesive set.
\end{enumerate}
If no individual holds the opinion $y_1$ at time $k_1$, then this situation is reduced to the case when there are $R-1$ distinct opinions initially, and the convergence is established by the pre-assumption. If the node set $\V$ at time $k_1$ is divided into two sets $\E_1$ and $\V\setminus \E_1$, then, after the update sequence $i_{11},\dots,i_{1k_1}$, nodes in $\E_1$ never switch their opinion from $y_1$ to any other opinion, while nodes in $\V\setminus \E_1$ never switch their opinions to $y_1$.

Let $A_2 = \{y_1,y_2\}$ and $B_2 = \{y_3,\dots,y_R\}$. We have shown that, after the update sequence $i_{1,1},\dots,i_{1,k_1}$, the nodes holding opinion $y_1$ will never change their opinions. Therefore, due to the weighted-median mechanism, for all the nodes in $\V\setminus \E_1$, it makes no difference to their opinion updates whether the nodes in $\E_1$ hold opinion $y_1$ or $y_2$. As the result, in terms of determining the behavior of the nodes in $\V\setminus \E_1$, we can take $y_1$ and $y_2$ as the same opinion. Following the same line of argument in the previous paragraph, there exists another update sequence $i_{21},\dots, i_{2k_2}$, right after the sequence $i_{1,1},\dots,i_{1,k_1}$, such that, after these two sequences of updates, all the nodes are partitioned into two sets $\E_2$ and $\V\setminus \E_2$, where $\E_2$ is the set of all the nodes holding either opinion $y_1$ or opinion $y_2$, and $\E_2$ is a maximal cohesive set.

Repeating the argument in the previous paragraph, we obtain the sets $\E_1,\dots, \E_{n-1}$, which are all maximal cohesive sets, and the entire update sequence $i_{1,1},\dots,i_{1,k_1},\dots, i_{n-1,1},\dots,i_{n-1,k_{n-1}}$. Define 
\begin{align*}
\V_1 & = \E_1,\\
\V_r & = \E_r\setminus \cup_{s=1}^{r-1} \E_s, \text{ for any }r=2,\dots,n-1,\\
\V_n & = \V \setminus \cup_{s=1}^{n-1} \E_s.
\end{align*}
The way we construct $\E_1\dots, \E_{n-1}$ implies that, after the update sequence $i_{1,1},\dots,i_{1,k_1},\dots,i_{n-1,1},\dots,i_{n-1,k_{n-1}}$, the system reaches a state in which, for any $r\in \{1,\dots, n\}$, all the nodes in $\V_r$ hold the opinion $y_r$ and will not switch to any other opinion. Therefore, for any initial condition $x(0)\in \mathbb{R}^n$, we have constructed an update sequence, along which the system reaches an equilibrium in finite time. This concludes the proof for the case of $R$ distinct initial opinions and thereby concludes the proof of statement~1*) by induction, which is equivalent to statement~1).

Now we proceed to prove statement~2).  If the only maximal cohesive set in $\G(W)$ is $\V$ itself, then, according to Lemma~\ref{lem:properties-cohesive-expansion}~4), the cohesive expansion of any cohesive set is $\V$ itself. Repeating the construction of the update sequence in the proof of statement~1*), we have that $\E_1$ in that proof is either an empty set or the cohesive expansion of a cohesive set, which is $\V$ itself under the condition of statement~2). If $\E_1$ is empty, then the situation is reduced to the case of $R-1$ distinct initial opinions. Following the same argument, $\E_2$ is either empty or $\V$ itself. As this argument goes on, will find some $r\in \{1,\dots,R\}$ such that $\E_r=\V$ while $\E_1,\dots, \E_{r-1}$ are all empty. That is, for any initial condition, there exists an update sequence, along which the system reaches an equilibrium where all the individuals hold the same opinion. This concludes the proof of statement~2).


Regarding statement~3), suppose there exists a maximal cohesive set $\M\subset \V$. We construct the set $X_0$ of initial conditions as
\begin{align*}
X_0 = \Big{\{}x(0)\in \mathbb{R}^n\,\Big|\, & \max_{j\in \M} x_j(0) < \min_{k\in \V\setminus \M}x_k(0),\textup{ or } \\
& \min_{j\in \M} x_j(0) > \max_{k\in \V\setminus \M}x_k(0)\Big{\}}.
\end{align*}
By definition, the set $X_0$ has non-zero Lebesgue measure in $\mathbb{R}^n$. Moreover, according to the weighted-median mechanism, starting from any $x(0)\in X_0$, the opinions of the nodes in $\M$ will always be lower (higher resp.) than the opinion of any node in $\V\setminus \M$, if $\max_{j\in \M} x_j(0) < \min_{k\in \V\setminus \M}x_k(0)$ ($\min_{j\in \M} x_j(0) > \max_{k\in \V\setminus \M}$ resp.). This concludes the proof of statement~3).  

Now we proceed to prove statement~4). By definition, $\hat{X}$ has zero Lebesgue measure in $\mathbb{R}^n$. Consider any $x(0)\in \mathbb{R}^n\setminus \hat{X}$. Since no new opinion is created along the weighted-median opinion dynamics, the system reaches consensus if and only if there exists a node $i$ and some time $T\in \mathbb{N}$ such that $x_j(T)=x_i(0)$ for any $j\in \V$. Since no individual other than $i$ initially holds the opinion $x_i(0)$, as indicated by Remark~\ref{remark:implication-decisive/indecisive}, in order for a node $j$ to adopt the opinion $x_i(0)$ at time $T$, there must exist at least one path on the influence network $\G(W)$ from $j$ to $i$ via only decisive links. According to the conditions in statement~4), $\Gdecisive(W)$ does not have a globally reachable node. Therefore, such node $i$ does not exist. This concludes the proof of statement~4).
\end{proof}

At the end of this subsection, we briefly compare the dynamical behavior of the weighted-median opinion dynamics with some widely-studied averaging-based models. Regarding the conditions for reaching consensus or dissensus, the French-DeGroot model and the Friedkin-Johnsen model lead to somewhat over-simplified outcomes: The system either almost surely achieves consensus or almost surely achieves dissensus (with respect to initial conditions), dependent on whether a single graph-theoretic condition is satisfied. Some other extensions of the French-DeGroot model, e.g., the biased assimilation model and the bounded-confidence model, are highly nonlinear such that the conditions for consensus/dissensus are too difficult to analyze.

On the other hand, as indicated by Theorem~\ref{thm:dynamics-WM}, the weighted-median opinion dynamics exhibit three different behaviors: almost-sure consensus, almost-sure dissensus, or having a non-zero probability of reaching dissensus. Which behavior the model exhibits depends on two important network structure properties: the presence of non-trivial maximal cohesive sets and the connectivity of the influence network with all the indecisive links removed. Therefore, the weighted-median opinion dynamics, despite its simplicity in form, leads to richer and yet mathematically tractable dynamical behavior.

\subsection{Further analysis: Determining almost-sure dissensus is NP-hard}
In the previous section, we establish the almost-sure finite-time convergence of the weighted-median opinion dynamics and provide a necessary and sufficient graph-theoretic condition for almost-sure consensus. A sufficient condition for almost-sure dissensus is also provided. In this subsection, we provide necessary and sufficient conditions for almost-sure convergence to dissensus equilibria, and prove that, given an arbitrary directed and weighted graph, determining whether these conditions hold is NP-hard. This result highlights the fact that median consensus yields significantly richer and more complex phenomena than average consensus, where most convergence questions can be solved by linear algebraic tools. 

The basic strategy is to relate our problem mentioned above to the monotone not-all-equal 3-satisfiability (NAE3SAT) problem, which is known to be NP-hard.
\begin{definition}[Monotone NAE3SAT]\label{def:NAE3SAT}
Given $m$ clauses $c_1,\dots, c_m$, where each clause $c_j$ is a set of three indices $\{k_j^1,k_j^2,k_j^3\}$ (A same index can appear in multiple clauses or multiple times in the same clause), find $n$ binary variables $x_1,\dots, x_n\in \{-1,1\}$ such that, for each clause $c_j$, $x_{k_j^1}$, $x_{k_j^2}$ and $x_{k_j^3}$ are not all equal.
\end{definition} 

The following result is well-known in previous literature~\cite{TJS:78}.

\begin{lemma}[Monotone NAE3SAT is NP-hard]\label{lem:NAE3SAT-NPhard}
The monotone NAE3SAT problem as described in Definition~\ref{def:NAE3SAT} is NP-hard unless P=NP.
\end{lemma}

As indicated by Theorem~\ref{thm:dynamics-WM}, the weighted-median opinion dynamics given by Definition~\ref{def:WM-op-dyn} almost surely achieves either a consensus equilibrium or a dissensus equilibrium in finite time, for any initial condition. The following proposition presents some equivalent statements for not achieving almost-sure dissensus. The proof is provided in Appendix~\ref{append:proof-prop:almost-sure-disagreement}.

\begin{proposition}[Almost-sure persistent dissensus]\label{thm:almost-sure-disagreement}
Given an influence network $\G(W)$ and the weighted-median opinion dynamics as in Definition~\ref{def:WM-op-dyn}, the following statements are equivalent:
\begin{enumerate}
\item There exists a set $X_0\subseteq \mathbb{R}^n$ with positive measure in $\mathbb{R}^n$ such that, for any $X(0)\in X_0$, there exists an update sequence, along which consensus is achieved in finite time;
\item There exists $x_0\in \mathbb{R}^n$ with $n$ distinct entries and an update sequence $\{i_1,\dots, i_T\}$ such that, along this update sequence, the trajectory $x(t)$ starting from $x(0)=x_0$ achieves consensus at time $T$;
\item There exists a ternary vector $y\in \{-1,0,1\}^n$, with $y_i=0$ for some $i\in \{1,\dots, n\}$ and $y_j\in \{-1,1\}$ for any $j\in \{1,\dots, n\}\setminus \{i\}$, and an update sequence $\{i_1,\dots, i_T\}$ such that the trajectory $y(t)$ starting from $y(0)=y$ reaches consensus on $0$ at time $T$, i.e., $y(T)=\vect{0}_n$. 
\end{enumerate}
The weighted-median opinion dynamics achieve persistent dissensus if and only if any of (i)-(iii) does not hold.
\end{proposition}
\smallskip

From statement~3) of Proposition~\ref{thm:almost-sure-disagreement} we know that, in order to check whether trajectories of the weighted-median opinion dynamics almost surely achieve persistent dissensus, we need to check only a finite set of initial conditions, whose entries take values from $\{-1,0,1\}$. However, in what follows, we show that, given any influence network $\G(W)$, determining whether statement~3) in Proposition~\ref{thm:almost-sure-disagreement} holds is an NP-hard problem. We tackle this problem in the following steps:
\begin{enumerate}[label={(\arabic*)}]
\item We show that, for any instance of the monotone NAE3SAT problem, we can build an associated influence network $\G(W)$ such that the NAE3SAT instance admits a solution if and only if the weighted-median opinion dynamics on $\G(W)$ do not almost surely achieve dissensus equilibria. Moreover, this construction can be done in polynomial time with respect to the size of the initial monotone NAE3SAET problem.
\item Step (1) above implies that, if we had an algorithm to determine in polynomial time whether the weighted-median opinion dynamics on any influence network $\G(W)$ achieves almost-sure dissensus, we could then combine it with the construction in Step (1) to create an algorithm solving monotone NAE3SAT in polynomial time. That is, our problem is solvable in polynomial time only if NAE3SAT is.
\item The NP-hardness of the monotone NAE3SAT, unless N=NP, implies then the NP-hardness of determining whether the weighted-median opinion dynamics achieve almost-sure dissensus.
\end{enumerate}
The key in the above steps is to construct a class of influence networks and relate the testing of statement~3) in Proposition~\ref{thm:almost-sure-disagreement} on these networks to an instance of the monotone NAE3SAT problem.

\begin{figure*}
\centering
\includegraphics[width=0.7\linewidth]{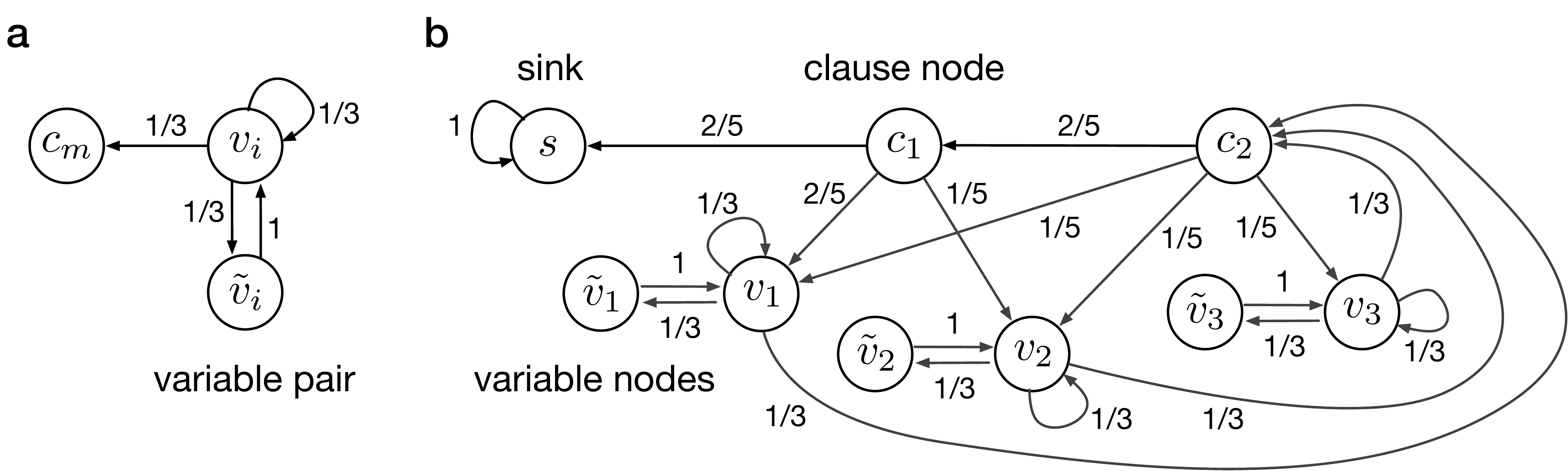}
\caption{Illustrations of the structure of the sink-variable-clause graphs. Panel a shows how a pair of variable nodes are connected to each other and to the last clause node. Panel b shows a simple example of an SVC graph with $m=2$ and $n=3$. In this graph, variable pairs $(v_1,\tilde{v}_1)$, $(v_2,\tilde{v}_2)$ are associated with clause node $c_1$, and variable pairs $(v_1,\tilde{v}_1)$, $(v_2,\tilde{v}_2)$, $(v_3,\tilde{v}_3)$ are associated with clause node $c_2$.}
\label{fig:svc-graph}
\end{figure*}

Given a monotone NAE3SAT problem, we construct the following specific class of influence networks and refer to them as ``sink-variable-clause'' (SVC) graphs.

\begin{definition}[Sink-variable-clause graphs]\label{def:SVCgraphs}
Given a monotone NAE3SAT problem with $n$ binary variables and $m$ clauses as in Definition~\ref{def:NAE3SAT}, a associated sink-variable-clause (SVC) graph is constructed as a weighted and directed graph with $2n+m+1$ nodes classified into 3 categories:
\begin{itemize}
\item the sink node with the index $s$;
\item the variable nodes that come in pairs ($n$ pairs in total). The indices for each pair are denoted by $(v_i,\tilde{v}_i)$, $i\in \{1,\dots,n\}$;
\item the clause nodes, with the indices $c_j$, $j\in \{1,\dots, m\}$. 
\end{itemize}
Nodes in the sink-variable-clause graph are connected with each other obeying the following rules:
\begin{enumerate}
\item The sink node has a self loop with weight 1 on itself;
\item For each $i\in \{1,\dots,n\}$, we have a pair of variable nodes $(v_i,\tilde{v}_i)$. Node $\tilde{v}_i$ assigns 100\% of its weight on node $v_i$. Node $v_i$ assigns a weight 1/3 to itself, a weight 1/3 to node $\tilde{v}_i$, and a weight 1/3 to the last clause node $c_m$, see Fig.~\ref{fig:svc-graph}\textbf{a};
\item Any clause node $c_j$ is associated with three variable nodes $k_j^1$, $k_j^2$, $k_j^3 \in \{v_1,\dots,v_n\}$, at most two of which might be the same. The clause node $c_j$ assigns a weight 1/5 to each of these 3 associated variable nodes (or 2/5 if an index appears twice), and a weight 2/5 to the previous clause node $c_{j-1}$ if $j>1$, or a weight 2/5 to the sink node $s$ if $j=1$, see Fig.~\ref{fig:svc-graph}\textbf{b};
\item Each variable node is linked to by at least one clause node. That is, $\{k_1^1,k_1^2,k_1^3,\dots, k_m^1,k_m^2,k_m^3\}=\{v_1,\dots,v_n\}$.
\end{enumerate}
\end{definition}

The idea of constructing the above SVC graph is as follows: Firstly, if the sink node starts with opinion $0$, it remains 0 all the time and this 0 shall propagate to the entire graph if the other nodes' initial opinions and update sequence are chosen properly; Secondly, each pair of variable nodes start with opposing opinions -1 and 1, and we keep them from changing their opinions until all the clause nodes have their opinions updated to 0; Thirdly, if the clause nodes start with opinions from $\{-1,\,1\}$ and update their opinions sequentially according to their indices $c_1,\, c_2,\, \dots, \, c_m$, their opinions will all become $0$, if and only if the variable nodes' opinions satisfy the rules of monotone NAE3SAT. Then the variable nodes' opinions can all be updated to 0 following some update sequence. The following lemma will show that certain of the properties mentioned above do indeed hold. This lemma is proved by directly applying Definition~\ref{def:WM-op-dyn}.

\begin{lemma}[Properties of SVC graphs]\label{lem:properties-SVC}
For any sink-variable-clause graph constructed as in Definition~\ref{def:SVCgraphs}, the following statements hold:
\begin{enumerate}
\item For any $i\in \{1,\dots, n\}$, if $x_{c_m}(t)=0$ and $x_{v_i}(t)=-x_{\tilde{v}_i}(t)\neq 0$ for some $t$, then an update of node $v_i$ followed by an update of node $\tilde{v}_i$ leads to $x_{v_i}(t+2)=x_{\tilde{v}_i}(t+2)=0$;
\item As long as $x_{v_i}(t)=x_{v_i}(0)\neq 0$ for any $i\in \{1,\dots, n\}$, node $c_j$ ($j\in \{1,\dots, m\}$) can update to 0 from time step $t$ to $t+1$ if and only if the following conditions are simultaneously satisfied:
   \begin{enumerate}
   \item $x_{c_{j-1}}(t)=0$ for $j>1$ and $x_s(t)=0$ for $j=1$. The latter holds as long as $x_s(0)=0$;
   \item $x_{k_j^1}(t)$, $x_{k_j^2}(t)$, and $x_{k_j^3}(t)$ are not all equal.
   \end{enumerate}
\end{enumerate}
\end{lemma} 

With all the preparation work above, now we state and prove the main theorem.
\begin{theorem}[Testing almost-sure dissensus is NP-hard]\label{thm:NP-hard}
Given an influence graph $\G(W)$, the problem of determining whether statement~3) in Proposition~\ref{thm:almost-sure-disagreement} holds is NP-hard. Hence, unless P=NP, no polynomial time algorithm can solve it.
\end{theorem}

\begin{proof}
Firstly of all, we point out that, according to Definition~\ref{def:SVCgraphs}, an SVC graph with any given $n$ and $m$ can be constructed in polynomial time.  

The main part of this proof is to show that, for any SVC graph with $2n+m+1$ nodes constructed according to Definition~\ref{def:SVCgraphs}, statement~3) in Proposition~\ref{thm:almost-sure-disagreement} hold if and only if there exists a solution to the associated monotone NAE3SAT with $m$ clauses and $n$ variables.

Regarding the ``if'' part, given any variable vector $x^*\in \mathbb{R}^{n}$ satisfying a monotone NAE3SAT problem with $m$ clauses, let $\G(W)$ be the associated SVC graph with $2n+m+1$ nodes. We now construct an initial condition $y(0)\in \mathbb{R}^{2n+m+1}$ and an update sequence, along which the weighted-median opinion dynamics over $\G(W)$ reach consensus on 0 in finite time. The initial condition $y(0)$ is constructed as follows:
\begin{align*}
y_s(0) & = 0,\\
y_{v_i}(0) & = x_{v_i}^*,\quad y_{\tilde{v}_i}(0)=-x_{v_i}^*,\quad \forall\, i=1,\dots,n,\\
y_{c_j}(0) & \in \{-1,1\},\quad \forall j=1,\dots, m.
\end{align*}
The update sequence is updated as follows: We first update the opinion of the clause node $c_1$. Note that $y_{k_1^1}(0)=x_{k_1^1}^*$, $y_{k_1^2}(0)=x_{k_1^2}^*$ and $y_{k_1^3}(0)=x_{k_1^3}^*$ are in the set $\{-1,1\}$ and are not all equal. Therefor, according to how the SVC graph $\G(W)$ is constructed, we have
\begin{align*}
\sum_{j:\, y_j(0)=-1} w_{c_1 j} &\le 2/5 < 1/2,\\
\sum_{j:\, y_j(0)=1} w_{c_1 j} &\le 2/5< 1/2,\quad\text{and}\\
\sum_{j:\, y_j(0)=0} w_{c_1 j} &= w_{c_1 s}=2/5 > 0.
\end{align*} 
That is, the opinion of $c_1$ will be updated to $0$ at time 1. Then we update the opinion of node $c_2$. Since $y_{k_2^1}(1)=y_{k_2^1}(0)=x_{k_2^1}^*$, $y_{k_2^2}(1)=y_{k_2^2}(0)=x_{k_2^2}^*$ and $y_{k_2^3}(1)=y_{k_2^3}(0)=x_{k_2^3}^*$ are in the set $\{-1,1\}$ and are not all equal, and since $y_{c_1}(1)=0$, according to how $\G(W)$ is constructed, we immediately have that the opinion of $c_2$ will be updated to 0 at time $2$. According to the same argument and following the update sequence $\{c_1,\,c_2,\, \dots, \, c_m\}$, we have
\begin{align*}
y_{c_1}(m)=\dots=y_{c_m}(m)=0.
\end{align*}
Now we update the opinions of the variable pairs in the following way. For each variable pair $\{v_i,\tilde{v}_i\}$, note that at time $m$ one of these two nodes' opinions is 1 and the others' is -1. Moreover, we have $y_{c_m}(m)=0$. Therefore, according to how $\G(W)$ is constructed, an update of node $v_i$'s opinion followed by an update of node $\tilde{v}_i$'s opinion will make their opinions both become 0. As the result, after $2n$ time steps, all the variable nodes' opinions will also be updated to $0$. In summary, for the SVC graph, we have constructed an initial condition $y(0)$ and an update sequence $\{c_1,\dots,c_m,v_1,\tilde{v}_1,\dots,v_n,\tilde{v}_n\}$, along which the weighted-median opinion dynamics reach consensus on the opinion 0 at time $2n+m$. This concludes the proof for the ``if'' part.

Regarding the ``only if'' part, suppose that there exist $y(0)\in \mathbb{R}^{2n+m+1}$ and an update sequence that together satisfy statement~3) in Proposition~\ref{thm:almost-sure-disagreement} on an SVC graph $\G(W)$. We now prove that the vector $x^*$, with $x_{k_i^j}^*=y_{k_i^j}(0)$ for any $i\in \{1,\dots,m\}$ and $j\in \{1,2,3\}$, must satisfy the rules for the associated monotone NAE3SAT.

For $y(0)$ and the update sequence given above that lead to consensus on opinion 0, we first conclude that $y_s(0)$ must be 0, since the sink node $s$ can never change its opinion. Then we have $y_i(0)\in \{1,-1\}$ for any $i\in \{1,\dots,2n+m+1\}\setminus \{s\}$. Denote by $T_m$ the first time step at which the opinion of $c_m$ is updated to 0. For any $i\in \{1,\dots, n\}$, according to how the out-links of the nodes $v_i$ and $\tilde{v}_i$ are constructed, we know that, up to time $T_m$, none of the opinions of $v_i$ and $\tilde{v}_i$'s out-neighbors in $\G(W)$ has ever been updated to 0. That is, we have $y_{v_i}(t)\in \{-1,1\}$ and $y_{\tilde{v}_i}(t)\in \{-1,1\}$ for any $t\le T_m$. Moreover, since $\{v_i,\tilde{v}_i\}$ forms a cohesive set in $G(W)$, we have that $y_{v_i}(t)=-y_{\tilde{v}_i}(t)\neq 0$ for any $t\le T_m$. Otherwise, there exists some $t_i\le T_m$ such that $y_{v_i}(t)=y_{\tilde{v}_i}(t)\neq 0$ for any $t\ge t_i$, which will contradict the pre-assumption that $y(t)$ eventually reaches consensus on opinion $0$. Furthermore, since $y_{v_i}(t)=-y_{\tilde{v}_i}(t)\neq 0$ remains true for all $t\le T_m$ and only one node can update its opinion at each time, we have that neither node $v_i$ nor node $\tilde{v}_i$ has ever updated their opinion before $T_m$. Therefore, for any $i\in \{1,\dots,n\}$ and any $t\le T_m$,
\begin{align*}
y_{v_i}(t)=y_{v_i}(0)=-y_{\tilde{v}_i}(0) = -y_{\tilde{v}_i}(t).
\end{align*}
Now we look at the node $c_m$. By assumption,node $c_m$'s opinion is updated to $0$ at time $T_m$, i.e., the weighted-median opinion for node $c_m$ at time $T_m-1$ is 0. In addition, since none of $y_{k_m^1}(T_m-1)$, $y_{k_m^2}(T_m-1)$ and $y_{k_m^3}(T_m-1)$ is 0, the following two statements must simultaneously hold:
\begin{enumerate}
\item $y_{k_m^1}(T_m-1)$, $y_{k_m^2}(T_m-1)$ and $y_{k_m^3}(T_m-1)$ are not all equal, which implies that $y_{k_m^1}(0)$, $y_{k_m^2}(0)$ and $y_{k_m^3}(0)$ are not all equal.
\item $y_{c_{m-1}}(T_m-1)=0$. Denote by $T_{m-1}$ the first time that $y_{c_{m-1}}(t)$ is updated to 0. Then we immediately have $T_{m-1}<T_{m}$.
\end{enumerate}
Apply the same argument above to the clause nodes $c_{m-1},\dots, c_1$ sequentially, we have $T_1<T_2<\dots <T_m$, and, moreover, $y_{k_j^1}(0)$, $y_{k_j^2}(0)$, $y_{k_j^3}(0)$ are not all equal for any $j\in \{1,\dots, m\}$. Let $x^*$ be such that 
\begin{align*}
x_{k_i^j}^* = y_{k_i^j}(0)
\end{align*}
for any $i\in \{1,\dots, m\}$ and any $j\in \{1,2,3\}$. Then $x^*$ satisfies the rules for the associated monotone NAE3SAT. This concludes the proof for the ``only if'' part. 

Now we have proved that for the specific class of sink-variable-clause graphs, testing statement~3) in Proposition~\ref{thm:almost-sure-disagreement} is equivalent to finding a solution to the monotone NAE3SAT problem, which is known to be NP-hard. Suppose for any graph with arbitrary topology, testing statement~3) in Proposition~\ref{thm:almost-sure-disagreement} is not NP-hard, then the testing method can also be applied to the specific class of sink-variable-clause graphs constructed, and thus the latter is not NP-hard, which leads to a contradiction. Therefore, in general, testing statement~3) in Proposition~\ref{thm:almost-sure-disagreement} for any arbitrary graph is NP-hard. 
\end{proof}

Note that Theorem~\ref{thm:NP-hard} does not exclude the possibility that, for some specific class of graphs, testing whether statement~3) in Theorem~\ref{thm:NP-hard} holds is not an NP-hard problem. 

\section{Conclusion}

In this paper, we conduct a thorough theoretical analysis of the rich dynamical behavior of the weighted-median opinion dynamics, previously exhibited via simulations in~\cite{WM-FB-GC-JH-FD:22}. We characterize the set of all the equilibria and establish the almost-sure finite-time convergence of the dynamics. A necessary and sufficient condition for the almost-sure convergence and a sufficient condition for persistent dissensus are given. All these important behavior of the weighted-median opinion dynamics are related to some delicate structures of the influence networks, such as the cohesive sets and decisive links. The rich dynamical behavior revealed by the theoretical analysis in this paper, together with other desirable properties discussed via experimental validation and numerical comparisons in~\cite{WM-FB-GC-JH-FD:22}, supports the weighted-median mechanism as a well-founded micro-foundation of opinion dynamics. In terms of future research directions, various meaningful extensions could be made by incorporating the weighted-median mechanism with previous important extensions to the French-DeGroot model, e.g., the continuous-time weighted-median model, the persistent attachment to initial opinions or the presence of negative weights characterizing antagonistic relations.

\appendices
\section{Proof of Lemma~\ref{lem:uniqueness-cohesive-expansion}}\label{append:proof-uniqueness-cohesive-expansion}
We prove the uniqueness of cohesive expansion by contradiction. For any cohesive set $\M\subset \V$, suppose that $\E_1 = \M\cup (i_1,\dots, i_k)$ and $\E_2 = \M\cup (j_1,\dots,j_{\ell})$ are two distinct cohesive expansions of $\M$. Here the ordered sets $(i_1,\dots, i_k)$ and $(j_1,\dots, j_l)$ indicate the orderings of node additions along the corresponding cohesive expansions. Since $\E_1\neq \E_2$ means that one of these two sets must contain at least one element that is not in the other, without loss of generality, suppose that there exists some $s_0\in \{1,\dots,l\}$ such that $j_{s_0}\notin (i_1,\dots, i_k)$. Now we prove that this is impossible. 

First of all, $s_0$ cannot be $1$, otherwise
\begin{align*}
\sum_{r\in \E_1}w_{j_1 r}\ge \sum_{r\in \M}w_{j_1 r}>1/2
\end{align*}
implies that $\E_1$ can be further expanded to $\E_1\cup (j_1)$. Secondly, there must exist $s_1\in \{1,\dots,s_0-1\}$ such that $j_{s_1}\notin (i_1,\dots,i_k)$, otherwise $\M\cup (j_1,\dots,j_{s_0-1})\subset \E_1$ and
\begin{align*}
\sum_{r\in \E_1} w_{j_{s_0}r} \ge \sum_{r\in \M\cup (j_1,\dots,j_{s_0-1})} w_{j_{s_0 }r}>1/2,
\end{align*}
which implies that $\E_1$ can be further expanded to $\E_1\cup (j_{s_0})$. As the same argument goes on, we will obtain that $j_1\notin (i_1,\dots,i_k)$. But we have already shown in this paragraph that $j_1 \notin (i_1,\dots,i_k)$ cannot be true. Therefore, $\E_1\neq \E_2$ cannot be true. This concludes the proof.

\section{Proof of Lemma~\ref{lem:properties-cohesive-expansion}}\label{append:proof-lem:properties-cohesive-expansion}
The proof of statement~1) is straightforward: For any $i\in \M_1\cup \M_2$, since either $i\in \M_1$ or $i\in \M_2$, we have
\begin{align*}
\sum_{j\in \M_1\cup \M_2}w_{ij}\ge \max \left\{ \sum_{j\in \M_1} w_{ij},\, \sum_{j\in \M_2} w_{ij} \right\}\ge \frac{1}{2}.
\end{align*}
Therefore, $\M_1\cup\M_2$ is cohesive.

Let $\textup{Expansion}(\M)=\M\cup (i_1,\dots,i_k)$, where the indices $1,\dots,k$ indicate the ordering of node additions along the cohesive expansion of $\M$. For any given $p\in \{1,\dots, k\}$, according to Definition~\ref{def:cohesive-expansion} and since $\M\subseteq \tilde{\M}$, we have
\begin{align*}
\sum_{j\in \tilde{M}\cup \{i_s:\, s<p\}} w_{i_p j}\ge \sum_{j\in \M\cup \{i_s:\, s<p\}} w_{i_p j}>\frac{1}{2}.
\end{align*}
Therefore, all of the nods $i_1,\dots, i_k$ are included in the cohesive expansion of $\tilde{M}$, which in turn implies $\textup{Expansion}(\M)\subseteq \textup{Expansion}(\tilde{\M})$. This concludes the proof of Statement~2).

According to Lemma~\ref{lem:uniqueness-cohesive-expansion}, since $\M\subseteq \M\cup \tilde{\M}$ and $\tilde{\M}\subseteq \M\cup \tilde{\M}$, we have  $\textup{Expansion}(\M)\subseteq \textup{Expansion}(\M\cup\tilde{\M})$ and $\textup{Expansion}(\tilde{\M})\subseteq \textup{Expansion}(\M\cup\tilde{\M})$. Therefore,  $\textup{Expansion}(\tilde{\M})\cup \textup{Expansion}(\M) \subseteq \textup{Expansion}(\M\cup\tilde{\M})$. This concludes the proof of statement~3).

With Statement~3), the proof of statement~4) becomes easy. Let $\tilde{\M}$ be any arbitrary maximal cohesive set such that $\M\subseteq \tilde{\M}$. According to how cohesive expansions are constructed as in Definition~\ref{def:cohesive-expansion}, we have $\textup{Expansion}(\tilde{M})=\tilde{M}$. Moreover, Statement~3) implies that 
\begin{align*}
\textup{Expansion}(\M)\subseteq \textup{Expansion}(\tilde{\M})=\tilde{M}.
\end{align*}
This concludes the proof of Statement~4).\hfill \qed

\section{Proof of Lemma~\ref{lem:monkey-typewriter}}\label{append:proof-lem:monkey-typewriter}
For any given $x(0)\in \mathbb{R}^n$, due to the definition of weighted-median, we have $x(t)\in \Omega = \{x_1(0),\dots,x_n(0)\}^n$ along any update sequence. Here $\Omega$ is a finite set of at most $n^n$ distinct elements. According to Definition~\ref{def:WM-op-dyn}, at any time $t+1$, one individual is uniformly randomly picked and update their opinion via the weighted-median mechanism. To put in mathematically, for any $x\in \Omega$,
\begin{equation}
\textup{Prob}[x(t+1)=x^{(i)}\,|\,x(t)=x] = 1/n
\end{equation}
for any $x^{(i)}\in \Omega$ satisfying $x^{(i)}_i = \Med_i(x;W)$ and $x^{(i)}_j = x_j$ for any $j\neq i$. Therefore, the weighted-median opinion dynamics is a Markov chain over the finite state space $\Omega$. This Markov chain has absorbing states, i.e., all the equilibria characterized in Theorem~\ref{thm:equilibrium-set}. Moreover, according to the assumption of this lemma, for any $x\in \Omega$, there exists at least one update sequence along which the trajectory $x(t)$ starting from $x$ reaches an equilibrium. Therefore, the weighted-median opinion dynamics is an absorbing Markov chain. According to Theorem~11.3 in the textbook~\cite{CMG-JLS:97}, $x(t)$ starting from $x(0)$ almost surely converges to an equilibrium. Since the stochastic process $x(t)$ is a finite-state Markov chain, $x(t)$ reaches an equilibrium almost surely in finite time.\hfill \qed

\section{Proof of Proposition~\ref{thm:almost-sure-disagreement}}\label{append:proof-prop:almost-sure-disagreement}

``2) $\Rightarrow$ 1)'': Suppose 2) is true, i.e., there exists $x_0\in \mathbb{R}^n$ with $n$ distinct entries and an update sequence $\{i_1,\dots, i_T\}$ such that, along this sequence, $x(t)$ starting from $x(t)=x_0$ achieves consensus at time $T$. Let $r=\min_{i\neq j} |x_{0,i}-x_{0,j}|$. Since the entries of $x_0$ are all distinct, we have $r>0$. Let $X_0=\big{\{} x\in \mathbb{R}^n\,\big|\, \lVert x-x_0 \rVert_{\infty}<r/2 \big{\}}$. For any $\tilde{x}_0\in X_0$, the entries of $\tilde{x}_0$ have the same ordering as the entries of $x_0$ in terms of their values. That is, if $x_0$ satisfies $x_{0,j_1}\le x_{0,j_2}\le \dots \le x_{0,j_n}$, where $\{j_1,\dots, j_n\}$ is a permutation of $\{1,\dots,n\}$, then $\tilde{x}_0$ also satisfies $\tilde{x}_{0,j_1}\le \tilde{x}_{0,j_2}\le \dots \le \tilde{x}_{0,j_n}$. According to the definition of weighted median, for any $x\in \mathbb{R}^n$, if $\Med_i(x;W)=x_j$ for some $j$, then $\Med_i(\tilde{x};W)=\tilde{x}_j$ as long as the entries of $\tilde{x}$ have the same ordering as those of $x$. That is, in the weighted-median opinion dynamics, only the ordering of the individual opinions matters. Therefore, for any $\tilde{x}_0\in X_0$, the trajectory $x(t)$ starting with $x(0)=\tilde{x}_0$ and along the update sequence $\{i_1,\dots, i_T\}$ also achieves consensus at time $T$. Therefore, statement ii) implies statement i). 

``1) $\Rightarrow$ 2)'': Suppose statement 1) is true, i.e., there exists $X_0\subseteq \mathbb{R}^n$ with positive measure in $\mathbb{R}^n$ such that, for any $x_0\in X_0$, there exists an update sequence, along which the trajectory $x(t)$ with $x(0)=x_0$ achieves consensus in finite time. Since $X_0$ has positive measure in $\mathbb{R}^n$, there must exists $x_0\in X_0$ such that the entries of $x_0$ are all distinct. For this $x_0$, there exists an update sequence $\{i_1,\dots, i_T\}$ such that $x(t)$ starting with $x(0)=x_0$ achieves consensus at time $T$. That is, statement 1) implies 2).

``2) $\Rightarrow$ 3)'': Given a vector $x$ satisfying the conditions in statement~2) and the corresponding update sequence $\{i_1,\dots, i_T\}$, since $x(t)\in \{x_1(0),\dots,x_n(0)\}^n$ for any $t\in \mathbb{N}$, there exists $i\in \V$ such that $x_j(T)=x_i$ for any $j$, i.e., $X(T)$ reaches consensus at the value $x_i$. For any $t\in \mathbb{N}$, define 
\begin{align*}
\V_1(t) & = \big\{j\in \V \,\big|\, x_j(t)<x_i  \big\},\\
\V_2(t) & = \big\{j\in \V\,\big|\, x_j(t)>x_i  \big\}.
\end{align*}
Now define $y\in \{-1,0,1\}^n$ as follows: For the index $i$ specified above, let $y_i= 0$. Let $y_j = -1$ for any $j\in \V_1(0)$ and $y_j = 1$ for any $j\in \V_2(0)$. Consider the solution $y(t)$ to the weighted-median opinion dynamics starting with $y(0)=y$ and along the update sequence $\{i_1,\dots, i_T\}$. Define   
\begin{align*}
\U_1(t) & = \big\{j\in \V \,\big|\, y_j(t)=-1  \big\},\\
\U_2(t) & = \big\{j\in \V \,\big|\, y_j(t)=1  \big\}.
\end{align*}
By definition, $\U_1(0)=\V_1(0)$ and $\U_2(0)=\V_2(0)$. Now we prove $\V_1(t)=\U_1(t)$ and $\V_2(t)=\U_2(t)$ for any $t\in \mathbb{N}$. Without loss of generality, suppose $i_1\in \V_1(0)$ (otherwise we can switch the definitions of $\V_1(t)$ and $\V_2(t)$). At time step 1, only node $i_1$ updates her opinion. Since
\begin{align*}
i_1\in &  \V_2(1)  \,\, \Leftrightarrow \,\, \sum_{j\in \V_2(0)}w_{i_1,j} > \frac{1}{2} \,\, \Leftrightarrow \,\, i_1\in \U_2(1);\\
i_1\in & \V_1(1)  \,\, \Leftrightarrow \,\, \sum_{j\in \V_1(0)} w_{i_1,j} \le \frac{1}{2} \,\, \Leftrightarrow \,\, i_1\in \U_1(1);\\
i_1\in & \V\setminus (\V_1(1)\cup \V_2(1))  \,\, \Leftrightarrow \,\,  x_{i_1}(1) = x_i\\
&\qquad \Leftrightarrow \,\, \sum_{j\in \V\setminus \V_1(0)} w_{i_1,j} > \frac{1}{2} \,\, \sum_{j\in \V_2(0)} w_{i_1,j} <\frac{1}{2} \\
&\qquad \Leftrightarrow \,\, y_{i_1}(1) = 0\,\, \Leftrightarrow \,\, i_1\in V\setminus (\U_1(1)\cup \U_2(1)),
\end{align*}
we have $\V_1(1)=\U_1(1)$ and $\V_2(1)=\U_2(1)$. Following the same argument for $t=2,3,\dots$, we conclude that $\V_1(t)=\U_1(t)$ and $\V_2(t)=\U_2(t)$ for any $t$. Therefore,
\begin{align*}
\V_1(T) = \V_2(T) \text{ is empty} \,\, \Leftrightarrow \,\, \U_1(T) = \U_2(T) \text{ is empty},
\end{align*}
that is, $y_j(T)=0$ for any $j\in V$. This concludes the proof for ``statement~2) $\Rightarrow$ statement~3)''.

``3) $\Rightarrow$ 2)'':  Given $y\in \{-1,0,1\}$ satisfying the conditions in statement~3) and given the corresponding update sequence $\{i_1,\dots, i_T\}$, define $\U_1(t)$ and $\U_2(t)$ in the same way as above. For any $x\in \mathbb{R}^n$ that satisfies the following conditions: 1) $x_i=0$; 2) $x_j<0$ for any $j\in \U_1(0)$; 3) $x_j>0$ for any $j\in \U_2(0)$; 4) All the entries of $x$ are distinct. Define $\V_1(t)$ and $\V_2(t)$ in the same way as in last paragraph. Following the same argument in the proof for ``2) $\Rightarrow$ 3)'', we obtain that $\V_1(T)=\U_1(T)$ and $\V_2(T)=\U_2(T)$ are all empty. That is, the vector $x$ we construct satisfies the conditions for $x_0$ in statement~2). This concludes the proof for ``3) $\Rightarrow$ 2)''.
\hfill\qed

Let's cite a paper~\cite{SRE-TB:15}

\bibliographystyle{IEEEtran}
\bibliography{alias,WM,Main,New}

\end{document}